\documentclass[preprint,12pt]{elsarticle}

\usepackage{amsmath,amsthm,amssymb,amsfonts}
\usepackage{verbatim}
\usepackage{mathtools}
\usepackage{enumerate}
\usepackage{graphicx,caption,subcaption,float}
\usepackage{tikz}
\usetikzlibrary{calc,shapes.geometric}
\usepackage{xcolor}
\usepackage{rotating}
\usepackage{tkz-graph}
\usepackage[all,cmtip]{xy}
\usepackage[toc]{appendix}
\usepackage{hyperref}
\usepackage{lineno}

\usepackage{mathrsfs}

\newcommand{\R}{\mathbb{R}}
\newcommand{\Rplus}{\mathbb{R}_{>0}}
\newcommand{\Rnn}{\mathbb{R}_{\geq 0}}
\newcommand{\invtPoly}{\mathcal{P}}
\newcommand{\St}{\mathscr{S}}

\DeclareMathOperator{\im}{im}

\newtheorem{theorem}{Theorem}[section]

\newtheorem{lemma}[theorem]{Lemma}
\newtheorem{proposition}[theorem]{Proposition}

\theoremstyle{definition}
\newtheorem{remark}[theorem]{Remark}
\newtheorem{example}[theorem]{Example}
\newtheorem{notation}[theorem]{Notation}
\newtheorem{condition}[theorem]{Condition}
\newtheorem{definition}[theorem]{Definition}

\DeclareSymbolFont{bbold}{U}{bbold}{m}{n}
\DeclareSymbolFontAlphabet{\mathbbold}{bbold}

\usepackage{amstext} 
\usepackage{array}   
\newcolumntype{L}{>{$}l<{$}} 

\newcommand{\rev}[2]{$\hspace{-2.6em}\xymatrix{
    &#1 \ar @<.4ex> @{-^>} [r]
    &\ar @{--^>} [l] #2
   }$}

\everymath{\displaystyle}

\journal{~}

\begin{document}

\begin{frontmatter}

\title{An All-Encompassing Global Convergence Result for Processive Multisite Phosphorylation Systems}
\author{Mitchell Eithun}
\address{Department of Mathematical Sciences, Ripon College, Ripon, WI 54971}
\ead{eithunm@ripon.edu}

\author{Anne Shiu}
\address{Department of Mathematics, Texas A\&M University, College Station, TX 77843-3368}
\ead{annejls@math.tamu.edu}

\date{June 2, 2017}

\begin{abstract}
Phosphorylation, the enzyme-mediated addition of a phosphate group to a molecule, is a ubiquitous chemical mechanism in biology. 
Multisite phosphorylation, the addition of phosphate groups to multiple sites of a single molecule, may be distributive or processive. Distributive systems, which require an enzyme and substrate to bind several times in order to add multiple phosphate groups,
can be bistable. Processive systems, in contrast, require only one binding to add all phosphate groups, and were recently shown to be globally stable. 
However, this global convergence result was proven only for a specific mechanism of processive phosphorylation/dephosphorylation (namely, all catalytic reactions are reversible). 
Accordingly, we generalize this result to allow for processive phosphorylation networks in which each reaction may be irreversible, and also to account for possible product inhibition. We accomplish this by first defining an all-encompassing processive network that encapsulates all of these schemes, and then appealing to recent results of Marcondes de Freitas, Wiuf, and Feliu that assert global convergence by way of monotone systems theory and network/graph reductions (corresponding to removing intermediate complexes).
Our results form a case study into the question of when global convergence is preserved when reactions and/or intermediate complexes are added to or removed from a network.  
\end{abstract}


\end{frontmatter}


\section{Introduction}

We address the question of when global dynamics, such as global convergence to a unique equilibrium, are preserved when reactions and/or intermediate complexes are added to or removed from a biochemical network.  Our work forms a case study into this question, by analyzing networks that describe the processive multisite phosphorylation/dephosphorylation of a molecule (a so-called ``multiple futile cycle'').  We now recall possible mechanisms underlying such a network.

\subsection{Mechanisms of Processive Multisite Phosphorylation} \label{sec:pro-lit}
A biological process of great importance, {\em phosphorylation} is the enzyme-mediated addition of a phosphate group to a protein substrate.  This process often modifies the function of the substrate.  The reactions underlying this mechanism are: $S_0 + E \leftrightharpoons S_0E \to S_1+E$, where $S_i$ is the substrate with $i$ phosphate groups attached and $E$ is the enzyme.

Additionally, many substrates have more than one {\em site} at which phosphate groups can be attached.  Such multisite phosphorylation may be {\em distributive} or {\em processive}, or somewhere in between~\cite{Guna_threshold, PM}.
In distributive phosphorylation, 
each binding of an enzyme to a substrate results in at most one addition 
of a phosphate group.
In contrast, in processive phosphorylation, 
when an enzyme catalyzes the addition 
of a phosphate group, phosphate groups are added 
to all sites before the enzyme and substrate dissociate.

Most studies on the mathematics of multisite phosphorylation have focused on phosphorylation under a sequential and fully {\em distributive} mechanism~\cite{MPM_MAPK,a6maya,KathaMulti,ManraiGuna,WangSontag}.  These systems admit bistability~\cite{bistable,Markevich} and oscillations~\cite{Errami}, and the set of steady states is parametrized by monomials~\cite{translated,TG,TSS}.

As for {\em processive} phosphorylation, Conradi and Shiu~\cite{Shiu} considered the following processive $n$-site phosphorylation/dephosphorylation network:
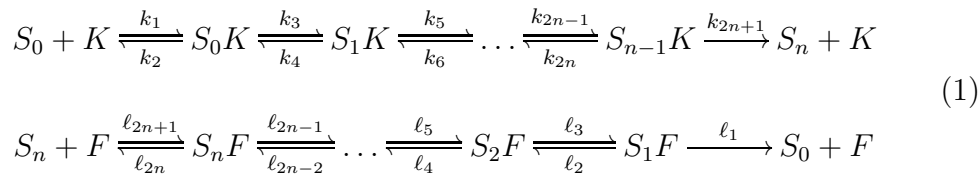
\begin{equation}
  \label{eq:multisite}
  \begin{split}
  \xymatrix{
    S_0 + K  \ar @<.4ex> @{-^>} [r] ^-{k_1}
    &\ar @{-^>} [l] ^-{k_{2}} S_0 K \ar @<.4ex> @{-^>} [r] ^-{k_3}
    &\ar @{-^>} [l] ^-{k_{4}} S_1 K \ar @<.4ex> @{-^>} [r] ^-{k_{5}}
    &\ar @{-^>} [l] ^-{k_{6}} \hdots \ar @<.4ex> @{-^>} [r] ^-{k_{2n-1}}
    &\ar @{-^>} [l] ^-{k_{2n}} S_{n-1} K \ar [r] 
    ^-{k_{2n+1}}
    & S_n + K\\
    S_n + F \ar @<.4ex> @{-^>} [r] ^-{\ell_{2n+1}}
    &\ar @{-^>} [l] ^-{\ell_{2n}} S_{n} F \ar @<.4ex> @{-^>} [r]
    ^-{\ell_{2n-1}} 
    &\ar @{-^>} [l] ^-{\ell_{2n-2}} \hdots \ar @<.4ex> @{-^>} [r]
    ^-{\ell_5}
    &\ar @{-^>} [l] ^-{\ell_4} S_2 F \ar @<.4ex> @{-^>} [r]
    ^-{\ell_3}
    &\ar @{-^>} [l] ^-{\ell_2} S_1 F \ar [r]
    ^-{\ell_1}
    &S_0 + F
  }
  \end{split}
\end{equation}
They proved that every resulting dynamical system (arising from mass-action kinetics), in contrast with distributive systems, does {\em not} admit bistability or oscillations, and, moreover, exhibits rigid dynamics.  Specifically, each invariant set (specified by conservation laws) contains a unique steady state, which is a global attractor~\cite{Shiu}. Conradi and Shiu proved this result via monotone systems theory, by generalizing a result of Angeli and Sontag \cite{AS}. Subsequently, using other means, Ali Al-Radhawi~\cite[\S 8.3]{ali2015new}, Rao \cite{Rao}, and Marcondes de Freitas, Wiuf, and Feliu \cite{Freitas2} established the same global convergence result.

However, in addition to~\eqref{eq:multisite}, there are other mechanisms for processive phosphorylation,  the following being the most common \cite{SK}: 
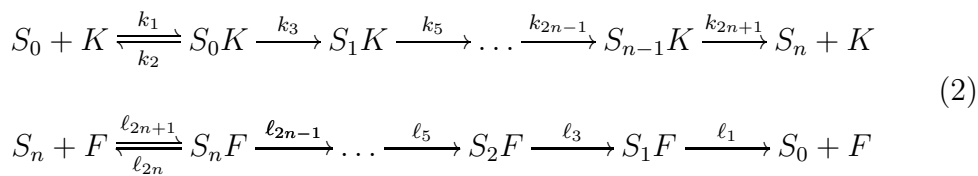
\begin{equation} 
\label{eq:irreversible}
\begin{split}
  \xymatrix{
    S_0 + K  \ar @<.4ex> @{-^>} [r] ^-{k_1}
    &\ar @{-^>} [l] ^-{k_{2}} S_0 K \ar [r] ^-{k_3}
    &S_1 K \ar [r] ^-{k_{5}}
    &\hdots \ar  [r] ^-{k_{2n-1}}
    & S_{n-1} K \ar [r] 
    ^-{k_{2n+1}}
    & S_n + K\\
    S_n + F \ar @<.4ex> @{-^>} [r] ^-{\ell_{2n+1}}
    &\ar @{-^>} [l] ^-{\ell_{2n}} S_{n} F \ar  [r] ^-{\ell_{2n-1}} 
    ^-{\ell_{2n-1}} 
    &\hdots \ar [r]
    ^-{\ell_5}
    &S_2 F \ar  [r]
    ^-{\ell_3}
    &S_1 F \ar [r]
    ^-{\ell_1}
    &S_0 + F
  }
\end{split}
\end{equation}
Here, in contrast with network~\eqref{eq:multisite}, the catalytic reactions are not reversible.

Another possible mechanism incorporates \textit{product inhibition}. 
Instead of detaching when the final phosphate group is attached or removed (e.g., $S_{n-1}K \to S_n +K$), the substrate and enzyme remain bound (e.g., $S_{n-1}K \to S_n K$), and then subsequently detach (e.g., $ S_n K \rightarrow S_n +K$).  Also, the final product (e.g., $S_n$) may rebind to the enzyme, thereby inhibiting its activity (e.g., $ S_n K \leftarrow S_n +K$).  Thus, a processive realization of this scheme is:
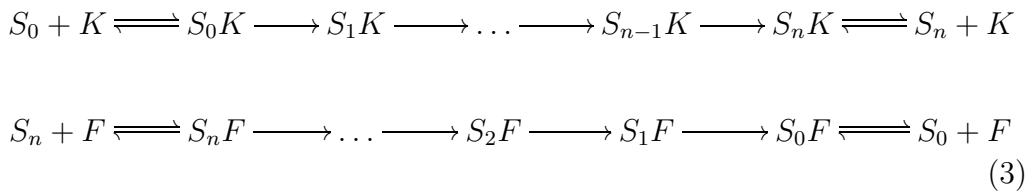
\begin{equation} 
\label{eq:inhibited}
\begin{split}
  \xymatrix{
    S_0 + K  \ar @<.4ex> @{-^>} [r]
    &\ar @{-^>} [l] S_0 K \ar  [r]
    &S_1 K \ar [r] 
    &\hdots \ar  [r] 
    &S_{n-1} K \ar[r] 
    &S_nK \ar @<.4ex> @{-^>} [r]
    &\ar @{-^>} [l] S_n + K 
    \\
    S_n + F \ar @<.4ex> @{-^>} [r] 
    &\ar @{-^>} [l]  S_{n} F \ar  [r] 
    &\hdots \ar  [r] 
    &S_2 F \ar [r] 
    &S_1 F \ar  [r]
    &S_0 F \ar @<.4ex> @{-^>} [r] 
    &\ar @{-^>} [l]  S_0 + F
  }
\end{split}
\end{equation}
There are distributive systems 
with such product inhibition \cite[Scheme 2]{Markevich}. 

Can the global stability result for \eqref{eq:multisite} be generalized to incorporate the other mechanisms (\ref{eq:irreversible}--\ref{eq:inhibited})? Indeed, we accomplish this in this work:

\begin{theorem} \label{thm:intro}
For any mass-action kinetics\footnote{In fact, other kinetics besides mass-action also work (see Remark~\ref{remark:rs}).} system arising from network \eqref{eq:multisite}, \eqref{eq:irreversible}, or~\eqref{eq:inhibited} and any choice of rate constants, each invariant set $\invtPoly$ contains a unique positive steady state and it is the global attractor of $\invtPoly$.
\end{theorem}
The proof of Theorem~\ref{thm:intro} appears in Section~\ref{sec:main-result}.  For now, we describe briefly the ideas behind the proof. 



\subsection{Proving Global Stability via an All-Encompassing Network}
To prove Theorem~\ref{thm:intro}, we construct an \textit{all-encompassing network} that subsumes all three networks~\eqref{eq:multisite}--\eqref{eq:inhibited}, and then prove the global convergence result for this network.  In this all-encompassing network, each reaction may be reversible or irreversible, there are $m$ reaction components rather than 2, and the number of binding sites in each component is allowed to differ.   

In addition to incorporating networks~\eqref{eq:multisite}--\eqref{eq:inhibited} as special cases, our all-encompassing network also specializes to 1-site phosphorylation networks (futile cycles) and certain cyclic networks introduced by Rao \cite{Rao}.  Hence, our global convergence result for the all-encompassing network generalizes prior global convergence results, including those of Angeli and Sontag \cite{AS} and Donnell and Banaji~\cite{DB} (for the 1-site network), Conradi and Shiu \cite{Shiu} and Marcondes de Freitas, Wiuf, and Feliu \cite{Freitas2} (network~\eqref{eq:multisite}), and Rao \cite{Rao}.

To prove our global convergence result, we use monotone systems theory and network/graph reductions.  
Specifically, we use a graph-theoretic criterion for global convergence from monotone systems theory.  This criterion, due to Angeli, De Leenheer, and Sontag~\cite{Angeli2010}, asserts that a given network is globally convergent if two graphs built from the network, the so-called $R$- and $SR$-graphs, satisfy certain properties. To apply this result efficiently, in light of the fact that our network has many intermediate complexes such as $S_0K$ and $S_n F$, we additionally use recent results that allow us to remove many of these intermediate complexes before applying the global-convergence criterion.  These results, due to Marcondes de Freitas, Wiuf, and Feliu \cite{Angeli2010, Freitas2}, state that if the convergence criterion holds after removing intermediate complexes, then the criterion also holds for the original network.

\subsection{Outline}  
The outline of our work is as follows.  Section~\ref{sec:background} defines reaction networks and their associated dynamical systems.  
Section~\ref{sec:all-encompassing} introduces the all-encompassing network, Section~\ref{sec:main-result} states the main global convergence result, and Section~\ref{proof2} provides the proof. In Section~\ref{sec:relation-global}, we mention other approaches to proving global stability, and, in Section~\ref{sec:relation}, we comment on how the systems analyzed in this work compare to other related phosphorylation systems. A discussion appears in Section~\ref{sec:discussion}. 
Finally, Appendix A explains how we check a technical detail, namely, bounded-persistence.

\subsection{Notation}
To aid the reader, we list in Table~\ref{tab:notation} the notation that we use, which will be defined beginning in the next section.
\begin{table}[]
\centering
\caption{Notation used in this work.}
\label{tab:notation}
\begin{tabular}{ll}
Notation & Definition \\
\hline
$\mathcal{S}$ & species set \\
$\mathcal{C}$ & complexes set \\
$\mathcal{R}$ & reactions set \\
$s$ & number of species \\
$r$ & number of reactions \\
$(\mathcal{S},\mathcal{C},\mathcal{R})$ & reaction network \\
$ \mathscr{S}$ & stoichiometric subspace \\
$\invtPoly$ & stoichiometric compatibility class \\
$G_{SR} = (V_{SR},E_{SR},L_{SR})$ & directed SR-graph \\
$G_{R} = (V_{R},E_{R},L_{R})$ & R-graph \\
$\mathcal{K}$ & an orthant cone
\end{tabular}
\end{table}

\section{Background} \label{sec:background}
This section describes how mass-action kinetics define a dynamical system from a chemical reaction network. Our setup is based on \cite{Shiu} and \cite{Freitas2}. 

\subsection{Chemical Reaction Networks}
As an example, consider the chemical reaction
\begin{equation} \label{example}
    \xymatrix{
      A+B  \ar [r] ^{\kappa} 
      & 3A+C 
    } 
\end{equation}

\noindent 
A {\em chemical reaction network} is a directed graph that comprises various reactions, such as~\eqref{example}.
The vertices $A+B$ and $3A+C$ are \textit{complexes}, which are linear combinations of individual \textit{species}.  The complex on the left side of a reaction is the \textit{reactant}, and the complex on the right side is the \textit{product}.   A species in a reactant (respectively, product) complex is a {\em reactant species} (respectively, {\em product species}).

An \textit{irreversible} reaction  is denoted by a unidirectional arrow ($\rightarrow$). 
A reaction with a double arrow, such as $X  \leftrightharpoons Y$ denotes a \textit{forward reaction} $X \to Y$  and a \textit{backward reaction} $Y \to X$. Together these reactions are called a $\textit{reversible}$ reaction. The parameter $\kappa$ is known as a \textit{rate constant}. 

More formally, a {\em chemical reaction network} with $s$ species is a triple $G = (\mathcal{S},\mathcal{C},\mathcal{R})$, which consists of:
\begin{enumerate}
\item a finite nonempty set of species $\mathcal{S} = \{ S_1, \dots, S_s \}$,
\item a set of complexes $\mathcal{C}$ of the form $y=(\alpha_1,\dots,\alpha_s)\in \mathbb{Z}^s_{\geq 0}$, representing the coefficients that form a linear combination of the species, and 
\item a set of reversible ($y \leftrightharpoons y'$) and irreversible ($y \rightarrow y'$) reactions $\mathcal{R}$. 
\end{enumerate}
For a reaction $y \leftrightharpoons y'$ or $y \to y'$, we call $y' - y$ the \textit{reaction vector}, which describes the net change in species. For instance, the reaction vector of the example reaction shown earlier in~\eqref{example} is $y_2-y_1 =(2,-1,1)$, which
means 
that with each occurrence of the reaction, two units of $A$ and one of $C$ are produced, while one unit of $B$ is consumed.

\begin{example} \label{ex:n=1}
\textit{Phosphorylation} is a chemical mechanism that adds a phosphate group to a molecule. The following network (called the ``futile cycle'') describes 1-site
  phosphorylation/dephosphorylation; it is the $n=1$ case of both networks~\eqref{eq:multisite} 
  and~\eqref{eq:irreversible}: 
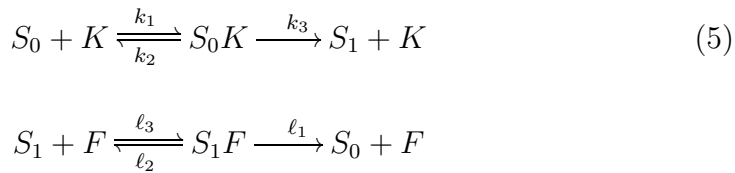
  \begin{equation}
    \label{eq:network_n=1}
    \xymatrix{
      S_0 + K  \ar @<.4ex> @{-^>} [r] ^{\;\;k_1} 
      &\ar @{-^>} [l] ^{\;\;k_2} S_0 K  \ar [r] 
      ^{k_{3}} 
      & S_1 + K\\ 
      S_1 + F \ar @<.4ex> @{-^>} [r] ^{\;\;\ell_3}
      &\ar @{-^>} [l] ^{\;\;\ell_{2}} S_{1} F \ar [r] 
      ^{\ell_{1}}
      &S_0 + F
    } 
  \end{equation}
  The key players in this network
  are a kinase ($K$), 
  a phosphatase ($F$), and a substrate ($S_0$).  
  The substrate $S_1$ is obtained from the unphosphorylated protein $S_0$ by attaching a phosphate group to it via an enzymatic reaction catalyzed by $K$. Conversely, a reaction catalyzed by $F$ removes the phosphate group from $S_1$ to obtain $S_0$. The intermediate complexes~$S_0 K$ and~$S_1 F$ are the  bound enzyme-substrate complexes.
\end{example}

\subsection{Mass-Action Kinetics}
Recall the example reaction $A+B \longrightarrow 3A+C$ from~\eqref{example}. Let $x_A, x_B$, and $x_C$ be the concentrations of the species as functions of time. Assuming the reaction follows \textit{mass-action kinetics}, the species $A$ and $B$ react proportionally to the product of their concentrations with constant of proportionality $\kappa$. Noting that the reaction yields a net change of two units in
the amount of $A$, we obtain the first differential equation in the following
system:
\begin{align*}
\frac{d}{dt}x_{A}~&=~2\kappa x_{A}x_{B}~ \\
\frac{d}{dt}x_{B} ~&=~-\kappa x_{A}x_{B}~ \\
\frac{d}{dt} x_{C}~&=~\kappa x_{A}x_{B}~.
\end{align*}
The other equations follow similarly. The mass-action differential equations defined by a network are a sum of monomial contributions, each of which corresponds to the reactant of a chemical reaction in the network.  These differential equations will be defined by equations~(\ref{eq:ODE}--\ref{eq:R-for-mass-action}).

Letting $r$ denote the number of reactions, where we count
each pair of reversible reactions only once, 
the {\em stoichiometric matrix}  
$\Gamma$ is the $s \times r$ matrix whose $k$-th column 
is the reaction vector of the $k$-th reaction 
 (the forward reaction if the reaction is reversible),
i.e., it is the reaction vector $y_j - y_i$ if $k$ indexes the (forward)
reaction $y_i \to y_j$.

The choice of kinetics is encoded by a locally Lipschitz function $R:\Rnn^s \to \R^r$ that lists the reaction rates of the $r$ reactions as functions of the $s$ species concentrations (a pair of reversible reactions is counted only once -- if the $k$-th reaction is reversible, then $R_k$ is the forward rate minus the backward rate).  
The {\em reaction kinetics system} 
defined by a reaction network $G$ and reaction rate function $R$ is given by the following system of ODEs:
\begin{align} \label{eq:ODE}
\frac{dx}{dt} ~ = ~ \Gamma \,  R(x)~.
\end{align}
For {\em mass-action kinetics}, the setting of this paper, the coordinates of $R$ are:
\begin{equation}\label{eq:R-for-mass-action}
 R_k(x)=\left\lbrace
 \begin{array}{ll}
  \kappa_{ij} x^{y_i} & \textrm{ if $k$ indexes an irreversible reaction $y_i \to y_j$} \\
  \kappa_{ij} x^{y_i} - \kappa_{ji} x^{y_j}  & \textrm{ if $k$ indexes a reversible reaction $y_i \leftrightharpoons y_j$} \\
 \end{array}\right. 
\end{equation}

A {\em chemical reaction system} refers to the 
dynamical system (\ref{eq:ODE}) arising from a specific chemical reaction
network $G$ and a choice of rate parameters $(\kappa_{ij}) \in
\mathbb{R}^{r}_{>0}$ (recall that $r$ denotes the number of
reactions) where the reaction rate function $R$ is that of mass-action
kinetics~\eqref{eq:R-for-mass-action}.


The {\em stoichiometric subspace} is the vector subspace of
$\mathbb{R}^s$ spanned by the reaction vectors
$y_j-y_i$ (where $y_i \to y_j$ is a reaction), and we denote this by: 
\begin{equation} \label{eq:stoic_subs}
  \St~:=~ {\rm span} \{ y_j-y_i\mid y_i \to y_j {\rm ~is~a~reaction~in~} G \}~.
\end{equation}
  Note that in the setting of (\ref{eq:ODE}), one has $\St = \im(\Gamma)$.
For example, the reaction vector $(2,-1,1)$ spans the
stoichiometric subspace $\St$ for the network~\eqref{example}. 

In general, the  vector $\frac{d x}{dt}$ in  (\ref{eq:ODE}) lies in
$\St$ for all time $t$.   
In fact, a trajectory $x(t)$ beginning at a positive vector $x(0) \in
\R^s_{>0}$ remains in the {\em stoichiometric compatibility class},
which we
denote by
\begin{align}\label{eqn:invtPoly}
\invtPoly~:=~(x(0)+\St) \cap \mathbb{R}^s_{\geq 0}~, 
\end{align}
for all positive time.  That is, $\invtPoly$ is forward-invariant with
respect to the dynamics~(\ref{eq:ODE}).    
A {\em positive steady state} of a kinetics system~\eqref{eq:ODE} is a positive concentration vector $x^* \in \R_{>0}^s$ at which the ODEs~\eqref{eq:ODE}  vanish: $\Gamma R(x^*) = 0$.  


\section{The All-Encompassing Network} \label{sec:all-encompassing}
Here we introduce a network that encompasses each of the three networks in the Introduction, and also encompasses a network introduced recently by Rao~\cite{Rao}. Accordingly, our new network has $m$ components rather than 2, each with its own enzyme $E_i$ and substrate $P_i$.  Also, each component has its own number of intermediate complexes and corresponding reactions. 

We let \rev{}{} denote a reaction that may or may not be reversible: it is either $\to$ or $\rightleftharpoons$.  The \textbf{all-encompassing} reaction network is:
\begin{equation}
\begin{split}
  \xymatrix@R-1pc{
    P_1 + E_1  \ar @<.4ex> @{-^>} [r] ^-{k_{11}}
    &\ar @{--^>} [l] ^-{k_{-11}} C_{11} \ar @<.4ex> @{-^>} [r] ^-{k_{12}}
    &\ar @{--^>} [l] ^-{k_{-12}} C_{12} \ar @<.4ex> @{-^>} [r] ^-{k_{13}}
    &\ar @{--^>} [l] ^-{k_{-13}} \;\; \hdots \;\;\; \ar @<.4ex> @{-^>} [r] ^-{k_{1n_{1}}}
    &\ar @{--^>} [l] ^-{k_{-1n_{1}}} \; C_{1n_{1}} \;\; \ar @<.4ex> @{-^>} [r] ^-{k_{1,n_1+1}}
    &\ar @{--^>} [l] ^-{k_{-1,n_1+1}} P_2 + E_1
    \\
    P_2 + E_2  \ar @<.4ex> @{-^>} [r] ^-{k_{21}}
    &\ar @{--^>} [l] ^-{k_{-21}} C_{21} \ar @<.4ex> @{-^>} [r] ^-{k_{22}}
    &\ar @{--^>} [l] ^-{k_{-22}} C_{22} \ar @<.4ex> @{-^>} [r] ^-{k_{23}}
    &\ar @{--^>} [l] ^-{k_{-23}} \;\; \hdots \;\;\; \ar @<.4ex> @{-^>} [r] ^-{k_{2n_{2}}}
    &\ar @{--^>} [l] ^-{k_{-2n_{2}}} \; C_{2n_{2}} \;\; \ar @<.4ex> @{-^>} [r] ^-{k_{2,n_2+1}}
    &\ar @{--^>} [l] ^-{k_{-2,n_2+1}} P_3 + E_2
    \\
   	\vdots&&&\vdots&&\vdots
    \\
    P_m + E_m  \ar @<.4ex> @{-^>} [r] ^-{k_{m1}}
    &\ar @{--^>} [l] ^-{k_{-m1}} C_{m1} \ar @<.4ex> @{-^>} [r] ^-{k_{m2}}
    &\ar @{--^>} [l] ^-{k_{-m2}} C_{m2} \ar @<.4ex> @{-^>} [r] ^-{k_{m3}}
    &\ar @{--^>} [l] ^-{k_{-m3}} \;\; \hdots \;\;\; \ar @<.4ex> @{-^>} [r] ^-{k_{mn_{m}}}
    &\ar @{--^>} [l] ^-{k_{-mn_{m}}} \; C_{mn_{m}} \;\; \ar @<.4ex> @{-^>} [r] ^-{k_{m,n_m+1}}
    &\ar @{--^>} [l] ^-{k_{-m,n_{m}+1}} P_1 + E_m
  }
\end{split}
\label{eq:m_network_processive}
\end{equation}
where $m \in \mathbb{Z}_{\geq 2}$ and $n_1,\dots,n_m \in \mathbb{Z}_{>0}$. As indicated, we allow each reaction to possibly be irreversible (in which case only the forward reaction takes place), that is, we impose the following restrictions on the rate constants:
$$k_{ij} > 0 \text{ and } k_{-ij} \geq 0 \text{ for all } i=1,\dots,m \text{ and } j=1,\dots,n_i.$$  This network has 
$2m + (n_1 + n_2 + \cdots + n_m)$ species. 

\begin{remark}
Technically, the all-encompassing network~\eqref{eq:m_network_processive} is not one network, but many -- one for each choice of $m$, $n_i$'s, and whether each reaction is reversible or irreversible.  Abusing notation, we nevertheless call it a network.
\end{remark}

\begin{remark} \label{rem:Rao} 
The all-encompassing network~\eqref{eq:m_network_processive} generalizes the network analyzed by Rao \cite{Rao}.  To obtain our network from Rao's, each reaction is allowed to be irreversible and the final reaction in each component may be reversible.  
Accordingly, the notation in \eqref{eq:m_network_processive} is based on Rao's \cite{Rao}, but with a few changes. Keeping with the convention that $n$ denotes the number of phosphorylation sites, $n_i$ denotes the number of intermediate complexes in component $i$, whereas Rao used the notation $m_i$ \cite{Rao}. Also, we use $m$ to represent the number of components in the network.
\end{remark}

Network~\eqref{eq:m_network_processive} 
generalizes not only Rao's network, but also the three mechanisms of processive phosphorylation/dephosphorylation in the Introduction:

\begin{proposition} \label{prop:encompasses-all}
The all-encompassing network~\eqref{eq:m_network_processive} includes as special cases, the processive multisite phosphorylation networks~\eqref{eq:multisite}, \eqref{eq:irreversible}, and \eqref{eq:inhibited}.
\end{proposition}

\begin{proof}
The conditions displayed here show how the all-encompassing network~\eqref{eq:m_network_processive} reduces to each of the three networks~\eqref{eq:multisite}, \eqref{eq:irreversible}, and \eqref{eq:inhibited}:
\begin{table}[H]
\begin{tabular}{ l  l }
	Network	& Conditions \\ 
  \hline
  \eqref{eq:multisite} & $m=2$, $n:=n_1=n_2$, $k_{-i,n+1}=0$ for $i=1,2$ \\
  \hline
  \eqref{eq:irreversible} & $m=2$, $n:=n_1=n_2$, $k_{-i,j}=0$ for $i=1,2$ and $j=2,\dots,n$ \\  
  \hline  
  \eqref{eq:inhibited}& $m=2$, $n+1:=n_1=n_2$, $k_{-i,j}=0$ for $i=1,2$ and $j=2,\dots,n$ \\  
  \hline  
\end{tabular}
\end{table}
\end{proof}


We end this section by showing that the all-encompassing network is conservative.

\begin{definition} \label{def:conservative}
A {\em positive conservation law} of a network $G$ is some $c \in {\rm ker}(\Gamma^T) \cap \mathbb{R}^s_{>0}$, where $\Gamma$ is the stoichiometric matrix of $G$ and $s$ is the number of species.  A network that has a positive conservation law is {\em conservative}.
\end{definition}

\begin{lemma} \label{cor:conservative}
The all-encompassing network \eqref{eq:m_network_processive} is conservative, and thus every one of its stoichiometric compatibility classes is compact. 
\end{lemma}

\begin{proof}
The vector $c \in \mathbb{R}^{2m + (n_1 + n_2 + \cdots + n_m)}_{>0}$, defined by $c_{P_i}:=1$, $c_{E_i}:=1$, and $c_{C_{ij}}:=2$ for all relevant $i$ and $j$, is a positive conservation law.  Every stoichiometric compatibility class is closed by construction and bounded due to the positive conservation law, and thus is compact.
\end{proof}

\section{Main Result: Global Convergence of All-Encompassing Network} \label{sec:main-result}

Our main result, which will be proven in Section \ref{sec:proof-subsec}, states that the all-encompassing network \eqref{eq:m_network_processive} is globally convergent:
\begin{theorem} \label{thm:main-result}
For any chemical reaction system~\eqref{eq:ODE} arising from the all-encompassing network~\eqref{eq:m_network_processive} and any choice of rate constants $k_{ij} > 0$ and $k_{-ij} \geq 0$,
\begin{enumerate}
	\item each compatibility class $\invtPoly$ contains a unique steady state $\eta$, 
	\item $\eta$ is a positive steady state, and 
	\item $\eta$ is the global attractor of $\invtPoly$.
\end{enumerate}
\end{theorem}

As a special case of Theorem~\ref{thm:main-result}, the three processive multisite phosphorylation networks from the Introduction are globally convergent:
\begin{proof}[Proof of Theorem~\ref{thm:intro}]
Follows from Proposition~\ref{prop:encompasses-all} and Theorem~\ref{thm:main-result}.
\end{proof}

Another special case of Theorem~\ref{thm:main-result} is Rao's result~\cite{Rao} (recall Remark~\ref{rem:Rao}). 
However, our proof differs from his (see Remark~\ref{rem:Rao2}).

\section{Proof of Main Result Using Reduced Networks/Graphs} \label{proof2} 
In this section, we prove Theorem~\ref{thm:main-result}.  To do so, we must recall 
how to construct two graphs from a chemical reaction network: the SR-graph and the R-graph. 
These graphs appear in the global convergence criterion from \cite{Angeli2010} that we will use.  Moreover, we will use a theorem from \cite{Freitas2} that allows us to first remove intermediate complexes to produce a reduced network, and then check the same graph-theoretic conditions on this simpler network. 

We recall the relevant setup and definitions in Sections~\ref{sec:assumptions}--\ref{sec:remove} and then state the relevant results from \cite{Freitas2} in Section~\ref{sec:results-from-freitas}. 
Accordingly, much of Sections~\ref{sec:assumptions}--\ref{sec:results-from-freitas} follow that in \cite{Angeli2010,Freitas2}.  Finally, our proof appears in Section~\ref{sec:proof-subsec}.


\subsection{Assumptions} \label{sec:assumptions}
In order for the results in \cite{Freitas2} to apply, a reaction network $(\mathcal{S}, \mathcal{C}, \mathcal{R})$ must satisfy the following assumptions\footnote{These assumptions do {\em not} limit the networks we can consider. Instead they clarify \textit{how} we represent networks.}:
\begin{enumerate}
\item for each complex $y \in \mathcal{C}$, there exists a reaction in $\mathcal{C}$ that has $y$ as a reactant or a product, and 
\item each species is contained in at least one complex.
\end{enumerate}

Some theorems in~\cite{Freitas2} additionally require the following conditions:

\noindent (G1) There are no auto-catalytic reactions, meaning that no species can be both a reactant species and a product species in any reaction. 

\noindent (G2) Each species in $\mathcal{S}$ takes part in at most two reactions in $\mathcal{R}$.

\noindent (G3) The network is {conservative} (recall Definition~\ref{def:conservative}). 

\begin{remark} \label{remark:rs}
The results in \cite{Freitas2} require assumptions on the choice of kinetics.  These assumptions, labeled (r1), (r2), and (r3), are satisfied by mass-action kinetics (such as our phosphorylation systems), power-law kinetics, and Hill kinetics~\cite[Remark 1]{Freitas2}, so they are omitted here.
\end{remark}

\subsection{The SR-graph and R-graph of a Reaction Network} \label{sec:graphs}
Here we explain how to construct two graphs from a chemical reaction network: the directed SR-graph and R-graph. Notationally, we write a directed, labeled graph as $G = (V,E,L)$, with vertex set $V$, edge set $E$, and labeling $L: E \to \{+,-\}$ (all edge labels here will be $+$ or $-$). 
A directed edge from $X$ to $Y$ is denoted by $\stackrel{\mathclap{\normalfont\mbox{$\rightarrow$}}}{XY}$.

A \textit{directed SR-graph}, denoted by $G_{SR} = (V_{SR},E_{SR},L_{SR})$, is a directed graph constructed from a chemical reaction network $(\mathcal{S},\mathcal{C},\mathcal{R})$ as follows. 
The vertex set $V_{SR}$ is the union of all species and reactions in the network (hence the name ``SR'').  
The edges and their labels are defined here\footnote{Our definitions for SR-graph and R-graph differ from those in~\cite{Freitas2}, but are equivalent.}:
\begin{enumerate}
\item If a species $S$ is a reactant species of a (reversible or irreversible) reaction $R \in \mathcal{R}$ or a product species of a reversible reaction $R \in \mathcal{R}$,  then $\stackrel{\mathclap{\normalfont\mbox{$\rightarrow$}}}{SR} \; \in E_{SR}$ and $\stackrel{\mathclap{\normalfont\mbox{$\rightarrow$}}}{RS} \; \in E_{SR}$.
\item If 
$S$ is a product species of an irreversible 
$R \in \mathcal{R}$,  then $\stackrel{\mathclap{\normalfont\mbox{$\rightarrow$}}}{RS} \; \in E_{SR}$.
\item Let $S$ be a species and $R$ a reaction.  If $S$ is a reactant species of $R$ (of the forward reaction of $R$ if $R$ is reversible), then 
	$L_{SR}(\stackrel{\mathclap{\normalfont\mbox{$\rightarrow$}}}{SR} ) := +$ and 
    $L_{SR}(\stackrel{\mathclap{\normalfont\mbox{$\rightarrow$}}}{RS} ) := +$. 
   If $S$ is a product species of $R$, then
       $L_{SR}(\stackrel{\mathclap{\normalfont\mbox{$\rightarrow$}}}{RS} ) := -$, and, if additionally $\stackrel{\mathclap{\normalfont\mbox{$\rightarrow$}}}{RS} \; \in E_{SR}$, then 
         $L_{SR}(\stackrel{\mathclap{\normalfont\mbox{$\rightarrow$}}}{SR} ) := -$. 		
\end{enumerate}

An \textit{R-graph} is an undirected graph $G_R = (V_R, E_R, L_R)$ created from a chemical reaction network
(in fact, from its directed SR-graph) as follows:
\begin{enumerate}
\item The vertex set $V_R$ is the set of reactions in the reaction network.
\item An edge connects reactions $R_i$ and $R_j$ if there is a length-2 path connecting $R_i$ and $R_j$ in the SR-graph. This edge is labeled with the opposite of the product of the two labels along the path. An edge may have more than one label, if there are multiple such paths.
\end{enumerate}

\begin{example}
Recall the 1-site phosphorylation system from Example \ref{ex:n=1}. The directed SR-graph and R-graph for this network are shown in Figure~\ref{ex:n=1_graphs}.

\begin{figure*}[h!]
\centering
\begin{subfigure}[t]{0.5\textwidth}
\centering

\begin{tikzpicture}

\node[draw] (R1) at (0,0) {$R_1$};
\node[draw] (R2) at (4,0) {$R_2$};
\node[draw] (R3) at (4,-2) {$R_3$};
\node[draw] (R4) at (0,-2) {$R_4$};
\node[] (S0) at (-1,-1) {$S_0$};
\node[] (E) at (2,0.5) {$K$};
\node[] (S0E) at (2,-0.5) {$S_0 K$};
\node[] (S1) at (5,-1) {$S_1$};
\node[] (S1F) at (2,-1.5) {$S_1F$};
\node[] (F) at (2,-2.5) {$F$};



\draw[<->,thick] (R1) to node[midway,above]{+} (S0);
\draw[<->,thick] (R1) to node[midway,above]{+} (E);
\draw[<->,thick] (R1) to node[midway,below]{-} (S0E);

\draw[->,thick] (R2) to node[midway,above]{-} (S1);
\draw[->,thick] (R2) to node[midway,above]{-} (E);
\draw[<->,thick] (R2) to node[midway,below]{+} (S0E);

\draw[<->,thick] (R3) to node[midway,below]{+} (S1);
\draw[<->,thick] (R3) to node[midway,above]{-} (S1F);
\draw[<->,thick] (R3) to node[midway,below]{+} (F);

\draw[->,thick] (R4) to node[midway,below]{-} (S0);
\draw[<->,thick] (R4) to node[midway,above]{+} (S1F);
\draw[->,thick] (R4) to node[midway,below]{-} (F);
\end{tikzpicture}
\caption{The directed SR-graph.}
\end{subfigure}%
    ~
\begin{subfigure}[t]{0.5\textwidth}
\centering
\begin{tikzpicture}
\node[draw] (R1) at (0,0) {$R_1$};
\node[draw] (R2) at (2,0) {$R_2$};
\node[draw] (dots) at (2,-2) {$R_3$};
\node[draw] (Rm) at (0,-2) {$R_4$};

\draw[thick]  (R1)  -- (R2) node[midway,above] {+};
\draw[thick]  (R2)  -- (dots) node[midway,right] {+};
\draw[thick]  (dots)  -- (Rm) node[midway,below] {+};
\draw[thick]  (Rm)  -- (R1) node[midway,left] {+};
\end{tikzpicture}
\caption{The R-graph.}
\end{subfigure}
\caption{The directed SR-graph and R-graph for the 1-site phosphorylation network.}
\label{ex:n=1_graphs}
\end{figure*}
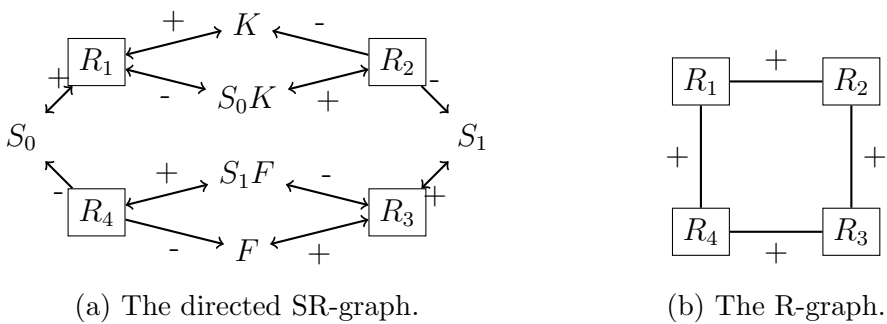

\end{example}

Next we define a property of an R-graph that, in Section~\ref{sec:results-from-freitas}, will help establish the global stability of a system. 

\begin{definition} ~
An R-graph has the {\em positive loop property} if every simple loop has an even number of negative edges.
\end{definition}

\begin{example}
Consider again the SR-graph and R-graph in Figure \ref{ex:n=1_graphs}. 
The R-graph has the positive loop property, because it has no negative labels.
\end{example}

\begin{notation} \label{rem:cone}
For a network whose R-graph has the positive loop property, we define an orthant cone (recall that $r$ denotes the number of reactions):
\begin{equation} 
\mathcal{K} := \{(x_1,\dots,x_r) \in \R^r \mid {\rm sign}(x_i) \in \{0, \sigma_i\} \text{ for all } i=1, \dots, r \},
\end{equation}
by defining a sign pattern $\sigma = (\sigma_1,\dots,\sigma_r) \in \{+,- \}^r$ as follows.  If the R-graph is connected, 
set $\sigma_1 := +$, and then for $i \in \{2,3,\dots, r\}$, choose any simple path $1 = i_0 \text{---} i_1 \text{---} \cdots \text{---} i_k = i$ in the R-graph from 1 to $i$, and define $\sigma_i$ to be the product of the labels along the path:
\begin{equation} \label{sign_pattern}
\sigma_i := \prod_{d=1}^k L_R(\{R_{i_{d-1}},R_{i_d}\}).
\end{equation}
The R-graph has the positive loop property, so every simple loop has an even number of negative edges, and thus $\sigma_i$ does not depend on the choice of path.

If the R-graph has more than one connected component, we apply the same procedure to each component, starting with $\sigma_i := +$ for the smallest index $i \in \{1,\dots,m\}$ such that $R_i$ belongs to that component.    
\end{notation}

\subsection{Removing Intermediates} \label{sec:remove}
As mentioned above, we will prove global stability via criteria on a network's SR-graph and R-graph.  These graphs are large in the case of the all-encompassing network, so we will use results in  \cite{Freitas2} (described in Section~\ref{sec:results-from-freitas}) that allow us to first simplify the network by removing intermediate complexes, before checking the required conditions on the simpler SR- and R-graphs. This removal procedure is described now.

\noindent
\begin{condition}[Conditions for removing an intermediate]
Let $G=(\mathcal{S},\mathcal{C},\mathcal{R})$ be a network with species set $\mathcal{S} = \{ S_1, S_2, \dots, S_s \}$.  
The \textit{support} of a complex $y = (\alpha_1,\dots,\alpha_s) \in \R^s_{\geq 0}$ is the set of constituent species of the complex:
$$\text{supp } y := \{S_i \in \mathcal{S} \mid \alpha_i > 0 \}.$$

For a complex $Y \in \mathcal{C}$, we define two conditions:
\begin{enumerate}
\item[(l1)] $Y$ consists of exactly one species which appears with coefficient 1 ($Y=S_i$ for some $i$) and does not appear in any other complex in the network.
\item[(l2)] There exist unique complexes $y = \alpha_1S_1 + \cdots + \alpha_s S_s$ and $y' = \alpha_1'S_1 + \cdots + \alpha_s'S_s$ such that the following hold:
\begin{enumerate}[(i)]
\item Either $y \to Y$ or $y \leftrightharpoons Y$ is a reaction in $\mathcal{R}$.
\item Either $Y \to y'$ or $Y \leftrightharpoons y'$ is a reaction in $\mathcal{R}$.
\item  Letting $\mathcal{E} := \text{supp } y \cap \text{supp } y'$ denote the set of common species of $y$ and $y'$, then $\sum_{S_i \in \mathcal{E}} \alpha_iS_i = \sum_{S_i \in \mathcal{E}} \alpha_i'S_i =: e$. 




\end{enumerate}
\end{enumerate}
\end{condition}

\begin{definition} \label{def:remove}
Given a network $G=(\mathcal{S},\mathcal{C},\mathcal{R})$ and a complex $Y \in \mathcal{C}$ that satisfies conditions (l1) and (l2), the reduced reaction network $G^*=(\mathcal{S}^*,\mathcal{C}^*,\mathcal{R}^*)$ obtained by \textbf{removing the intermediate} $Y$ is as follows. First, $\mathcal{R^*} := \mathcal{R}_c^* \cup \mathcal{R}_Y^*$, where $\mathcal{R}_Y^*$ is the subset of reactions in $\mathcal{R}$ that do not have $Y$ as a product or reactant, and 
\begin{equation} \label{eq:r-star-y}
\mathcal{R}^*_Y := \begin{cases}
\{y-e \leftrightharpoons y' - e\}, \; \text{if } y \leftrightharpoons Y \in \mathcal{R} \text{ and } Y \leftrightharpoons y' \in \mathcal{R} \\
\{y-e \rightarrow y' - e\}, \; \text{if } y \rightarrow Y \in \mathcal{R} \text{ or } Y \rightarrow y' \in \mathcal{R} \\
\end{cases}.
\end{equation}
Next, $\mathcal{C}^*$ is the set of reactant and product complexes of the reactions in $\mathcal{R^*}$.  Finally, $\mathcal{S}^*$ is the set of species that appear in at least one complex in $\mathcal{C^*}$.
\end{definition}

This procedure removes one intermediate. Any number of intermediates may be removed successively if conditions (l1) and (l2) are met at each step.

\begin{example} \label{ex:1-site-reduced}
Consider the 1-site phosphorylation network \eqref{eq:network_n=1}. Taking $S_0 + K$ and $S_1+K$ to be the unique complexes $y$ and $y'$ required by (l2), we can remove the intermediate $S_0K$, producing the reduced network:
  
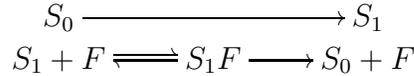
\begin{equation*} \label{ex:n=1_reduced}
\begin{split} 
  \xymatrix@R-2pc{
	&S_0 \ar [rr]
    && S_1
    \\
    &S_1 + F \ar @<.4ex> @{-^>} [r]
    & \ar @{-^>} [l] S_1F \ar [r]
    & S_0 + F
   }
\end{split}
\end{equation*}
Notice that $K$ is also removed, because it is in both $S_0 + K$ and $S_1+K$. 
\end{example}

The next lemma uses successive removal of intermediates to simplify the all-encompassing network. Recall that \rev{}{} denotes a reaction that may or may not be reversible.


\begin{lemma} \label{lem:removals}
The following network can be obtained from the all-encompassing network \eqref{eq:m_network_processive} by successive removal of intermediates:
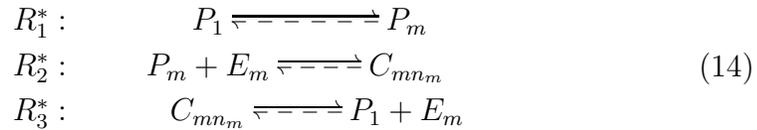
\begin{equation} \label{reduced}
\begin{split}
  \xymatrix@R-2pc{
    R_1^* :
    &P_1 \ar @<.4ex> @{-^>} [r]
    &\ar @{--^>} [l] P_m
    \\
    R_2^* :
    &P_m + E_m \ar @<.4ex> @{-^>} [r]
    &\ar @{--^>} [l] C_{mn_m}
    \\
    R_3^* :
    &C_{m n_m} \ar @<.4ex> @{-^>} [r]
    &\ar @{--^>} [l] P_1 + E_m
   }
\end{split}
\end{equation}

\end{lemma}
\begin{proof}
First, it is straightforward to check that for $i=1,2,\dots, m$, we can successively remove the intermediates $C_{i1}, C_{i2}, \dots, C_{i,n_i-1}$ from the all-encompassing network.  
The resulting network is:
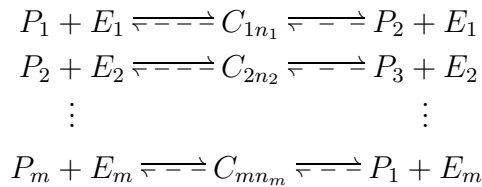
\begin{equation*} 
\begin{split} 
  \xymatrix@R-2pc{
    P_1 + E_1 \ar @<.4ex> @{-^>} [r]
    &\ar @{--^>} [l] C_{1n_1} \ar @<.4ex> @{-^>} [r]
    & \ar @{--^>} [l]  P_2 + E_1 
	\\
	P_2 + E_2 \ar @<.4ex> @{-^>} [r]
    &\ar @{--^>} [l] C_{2n_2} \ar @<.4ex> @{-^>} [r]
    & \ar @{--^>} [l]  P_3 + E_2 
    \\
    \vdots&&\vdots
    \\
    P_m + E_m \ar @<.4ex> @{-^>} [r]
    &\ar @{--^>} [l] C_{mn_m} \ar @<.4ex> @{-^>} [r]
    & \ar @{--^>} [l]  P_1 + E_m 
   }
\end{split}
\end{equation*}
Next, the intermediates $C_{1n_1}, C_{2n_2}, \dots, C_{m-1,n_{m-1}}$ can be removed (and at each step the corresponding $E_i$ as well, as per \eqref{eq:r-star-y}), which results in:
%
%
%
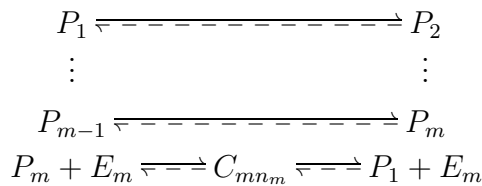
\begin{equation*} 
\begin{split} 
  \xymatrix@R-2pc{
    P_1 \ar @<.4ex> @{-^>} [rr]
    &&\ar @{--^>} [ll] P_2
    \\
    \vdots&&\vdots
    \\
     P_{m-1} \ar @<.4ex> @{-^>} [rr]
    &&\ar @{--^>} [ll] P_m 
    \\         
    P_m + E_m \ar @<.4ex> @{-^>} [r]
    &\ar @{--^>} [l] C_{mn_m} \ar @<.4ex> @{-^>} [r]
    & \ar @{--^>} [l]  P_1 + E_m 
   }
\end{split}
\end{equation*}


If $m=2$, we are done. Otherwise, we 
successively remove $P_2,P_3,\dots,P_{m-1}$. 
This results in the desired network \eqref{reduced}.
%
\end{proof}

\begin{remark}
Network \eqref{reduced} in Lemma \ref{lem:removals} can be reduced further, by removing the last intermediate $C_{n,n_m}$ to obtain the network:
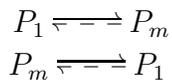
\begin{equation*}
\begin{split}
  \xymatrix@R-2pc{
    &P_1 \ar @<.4ex> @{-^>} [r]
    &\ar @{--^>} [l] P_{m}
    \\
    &P_m \ar @<.4ex> @{-^>} [r]
    &\ar @{--^>} [l] P_1
   }
\end{split}
\end{equation*}
However, when both of these reactions are reversible, then, following~\cite{Freitas2}, we would need to view the network as having two copies of the same (reversible) reaction.  To avoid this complication, we will use network~\eqref{reduced}.
\end{remark}

%

\subsection{Stability Results From~\cite{Angeli2010,Freitas2}} \label{sec:results-from-freitas}
To state results from~\cite{Angeli2010,Freitas2} that we will use, we need some definitions:

\begin{definition} \label{def:bounded_persistent} ~
\begin{enumerate}
\item The \textbf{$\omega$-limit set} of a trajectory $\sigma(t,s_0)$ of \eqref{eq:ODE} with initial condition $s_0$ is its set of limit points:
$$\omega(s_0) := \bigcap_{\tau \gg 0} \overline{\bigcup_{t \gg \tau} \{\sigma(t,s_0)\}}.$$
\item A network $G$ with $s$ species is {\bf bounded-persistent} if for all chemical reaction systems arising from $G$ and for all initial conditions $s_0 \in \mathbb{R}^s_{>0}$, the {$\omega$-limit set} of the resulting trajectory does {\em not} meet the boundary of the nonnegative orthant: $\omega(s_0) \cap \partial \R^{s}_{\geq 0} = \emptyset$. 
\end{enumerate}
\end{definition}

The following proposition follows directly from results of Marcondes de Freitas, Wiuf, and Feliu~\cite[Theorems 1--2]{Freitas2} and (as summarized in~\cite[Proposition 3]{Freitas2}
and Remark 6) Angeli, De Leenheer, and Sontag \cite{Angeli2010}.

\begin{proposition} \label{prop:summary-stability-results}f
Let $G$ be a reaction network satisfying (G1)--(G3), and let $G^*$ be a reaction network obtained from $G$ by successive removal of intermediates. Let $\Gamma^*$ be the stoichiometric matrix of $G^*$.   Assume that: 
\begin{enumerate}
\item $G$ is bounded-persistent, 
\item the R-graph of $G^*$ is connected and
has the positive loop property (so, from Notation \ref{rem:cone}, we can let $\mathcal{K}^*$ be the orthant cone constructed from this R-graph), and 
\item ${\rm ker }(\Gamma^*) \cap {\rm int } (\mathcal{K}^*) \neq \emptyset$, where ${\rm int } (\mathcal{K}^*)$ is the relative interior of $\mathcal{K}^*$.
\end{enumerate}
Then for the chemical reaction system\footnote{In fact, other kinetics besides mass-action also work (recall Remark~\ref{remark:rs}).} arising from $G$ and any choice of rate constants, each compatibility class $\invtPoly$ contains a unique steady state $\eta$, this steady state $\eta$ is a positive steady state, and $\eta$ is the global attractor of $\invtPoly \cap \mathbb{R}^s_{>0}$, where $s$ is the number of species.
\end{proposition}
Appendix \ref{sec:bounded_persistence} shows how bounded-persistence can be established with graph-theoretic criteria from \cite{Angeli2010}. Hence, each condition of Proposition~\ref{prop:summary-stability-results} is a graph-theoretic criterion (for the networks we are interested in).

Also, note that Proposition~\ref{prop:summary-stability-results} yields a global attractor of $\invtPoly \cap \mathbb{R}^s_{>0}$, not all of $\invtPoly$, so Appendix \ref{sec:bounded_persistence} contains a result that we will use to circumvent this.

\subsection{Proof of Global Stability of the All-Encompassing Network} \label{sec:proof-subsec}
\begin{proof}[Proof of Theorem~\ref{thm:main-result}]
Fix rate constants and a stochiometric compatibility class $\invtPoly$.  For any initial condition $x_0 \in \invtPoly$, the $\omega$-limit set $\omega(x_0)$ is a nonempty subset of $\invtPoly$ (because $\invtPoly$ is compact by Lemma~\ref{cor:conservative}) that does not intersect the boundary $\partial \invtPoly$ (by Lemma~\ref{lem:bounded-persistent}).  Thus, it suffices to show that there is a positive steady state in $\invtPoly$ that is a global attractor of $\invtPoly \cap \mathbb{R}^{2m + (n_1 + n_2 + \cdots + n_m)}_{>0}$.  

Accordingly, it is enough to show that the hypotheses of Proposition~\ref{prop:summary-stability-results} hold, where we take $G$ to be the all-encompassing network~\eqref{eq:m_network_processive} and $G^*$ to be the reduced network~\eqref{reduced} from Lemma~\ref{lem:removals}.  
We already know that $G$ is bounded-persistent (Lemma \ref{lem:bounded-persistent}), so we must show that 
(1) $G$ satisfies (G1)--(G3), 
(2) the R-graph of $G^*$ is connected,
(3) the R-graph of $G^*$ has the positive loop property, and 
(4) ${\rm ker }(\Gamma^*) \cap {\rm int } (\mathcal{K}^*) \neq \emptyset$. 

By inspection, $G$ satisfies (G1)--(G2). By Lemma \ref{cor:conservative}, $G$ satisfies (G3).
%


Figure \ref{fig:graphs} displays the SR- and R-graphs of the reduced network $G^*$. 
The R-graph is connected, so property (2) holds.
\begin{figure*}[h!] 
\centering
\begin{subfigure}[t]{0.5\textwidth}
\centering
\begin{tikzpicture}
\node[draw] (R1) at ($ ({2*cos(210)},{2*sin(210)}) $) {$R_1^*$};
\node[draw] (R2) at ($ ({2*cos(90)},{2*sin(90)}) $) {$R_2^*$};
\node[draw] (R3) at ($ ({2*cos(330)},{2*sin(330)}) $) {$R_3^*$};


\node (Pm1) at ($ ({2*cos(270)},{2*sin(270)}) $) {$P_1$};
\node (P1) at ($ ({2*cos(150)},{2*sin(150)}) $) {$P_m$};
\node (Pm) at ($ ({2*cos(30)},{2*sin(30)}) $) {$C_{mn_m}$};
\node (Em) at (0,0) {$E_m$};

\path[->,thick] (R2) edge [bend left,looseness=0.5,out=10,in=170] (Em);
\path[<-,thick] (R2) edge [bend right,looseness=0.5,out=-10,in=190] (Em);

\path[->,thick] (R3) edge [bend left,looseness=0.5,out=10,in=170] (Em);
\path[<-,thick,dashed] (R3) edge [bend right,looseness=0.5,out=-10,in=190] (Em);

\path[->,thick] (R1) edge [bend left,looseness=0.5] (P1);
\path[<-,thick,dashed] (R1) edge [bend right,looseness=0.5] (P1);
\path[->,thick] (R1) edge [bend left,looseness=0.5] (Pm1);
\path[<-,thick] (R1) edge [bend right,looseness=0.5] (Pm1);

\path[->,thick] (R2) edge [bend left,looseness=0.5] (P1);
\path[<-,thick] (R2) edge [bend right,looseness=0.5] (P1);
\path[->,thick] (R2) edge [bend left,looseness=0.5] (Pm);
\path[<-,thick,dashed] (R2) edge [bend right,looseness=0.5] (Pm);

\path[->,thick] (R3) edge [bend left,looseness=0.5] (Pm);
\path[<-,thick] (R3) edge [bend right,looseness=0.5] (Pm);
\path[->,thick] (R3) edge [bend left,looseness=0.5] (Pm1);
\path[<-,thick,dashed] (R3) edge [bend right,looseness=0.5] (Pm1);

\draw[draw=none]  (R1)  -- (P1) node[midway,above=-6] {$-$};
\draw[draw=none]  (R1)  -- (Pm1) node[midway,above=-6] {$+$};
\draw[draw=none]  (R2)  -- (P1) node[midway,above=-6] {$+$};
\draw[draw=none]  (R2)  -- (Pm) node[midway,above=-6] {$-$};
\draw[draw=none]  (R3)  -- (Pm) node[midway,above=-6] {$+$};
\draw[draw=none]  (R3)  -- (Pm1) node[midway,above=-6] {$-$};
\draw[draw=none]  (R2)  -- (Em) node[midway,left] {$+$};
\draw[draw=none]  (R3)  -- (Em) node[midway,below] {$-$};
\end{tikzpicture}

\caption{directed SR-graph}
\end{subfigure}%
    ~
\begin{subfigure}[t]{0.5\textwidth}
\centering
\begin{tikzpicture}
\node[draw,thick] (R1) at (0,0) {$R_1^*$};
\node[draw,thick] (R2) at (1,1.732) {$R_2^*$};
\node[draw,thick] (R3) at (2,0) {$R_3^*$};
\draw[]  (R1)  -- (R2) node[midway,left] {+};
\draw[]  (R2)  -- (R3) node[midway,right] {+};
\draw[]  (R3)  -- (R1) node[midway,below] {+};
\end{tikzpicture}
\caption{R-graph.}
\end{subfigure}
\caption{The SR-graph and R-graph of the reduced network \eqref{reduced}.  A dashed edge in the SR-graph is present if and only if every reaction in the corresponding component in the original all-encompassing network \eqref{eq:m_network_processive} is reversible.}
\label{fig:graphs}
\end{figure*}
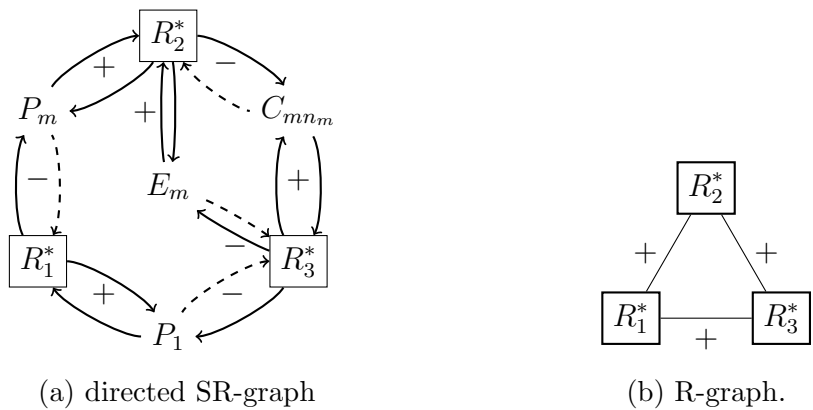
%

In the SR-graph, each length-two path connecting two reaction vertices consists of two edges with opposite signs. 
Thus, the R-graph has only edges with + labels, 
and so vacuously has the positive loop property (property (3)). 

Finally, because all edges of the R-graph 
are labeled by +, it follows that $\mathcal{K}^* = \mathbb{R}^3_{\geq 0}$.  Also, each species in $G^*$ appears in exactly two reactions, once as a reactant and once as a product, and so the sum of each row of $\Gamma^*$ is 0. 
Thus, $(1,1,1) \in {\rm ker }(\Gamma^*) \cap {\rm int } (\mathcal{K}^*)$, so property (4) holds.
%
\end{proof}

\section{Relation to Other Approaches to Proving Global Stability} \label{sec:relation-global}
Our method for proving global stability, via monotone systems theory, is only one of several approaches for proving stability of reaction systems (reviewed in \cite[\S 2.2]{ali2015new}).  Here we note two alternate approaches.

\begin{remark} \label{rem:ali}
For the processive network~\eqref{eq:multisite}, Ali Al-Radwahi gave a Lyapunov function~\cite[\S 8.3]{ali2015new}.  As for the all-encompassing network, Ali Al-Radwahi and Angeli's results again yield a Lyapunov function, which is piecewise linear in the reaction rate functions~\cite{ali-angeli}, thereby obtaining the same global convergence result as ours. Specifically, their Theorem~13 applies to the fully irreversible version of the all-encompassing network, and then their Theorem~14 applies when {\em any} of the reactions are made reversible.
\end{remark}

\begin{remark} \label{rem:Rao2}
As we noted in Remark~\ref{rem:Rao}, our Theorem \ref{thm:main-result} generalizes Rao's recent stability result~\cite{Rao}.
Like Ali Al-Radwahi and Angeli,
Rao built a Lyapunov function that is piecewise linear in the reaction rate functions. 

It appears that Rao's proof can be extended as another means to establish a version of Theorem \ref{thm:main-result} (S.\ Rao, personal communication).  There are, however, two important caveats: Rao's result applies only to mass-action kinetics, and the uniqueness of steady states in each compatibility class must be proven separately.
\end{remark}

\section{Relation to Other Multisite Phosphorylation Systems} \label{sec:relation}
Here we discuss how the phosphorylation networks analyzed in this work compare to others in the literature.

\begin{remark}
There are examples in the literature of processive phosphorylation networks that have more reactions than those in our all-encompassing network \eqref{eq:m_network_processive}. For instance, Gunawardena proposed a processive 2-site phosphorylation network in which 
$ES_0$ reacts to form $E + S_2$, as follows~\cite{Guna}:
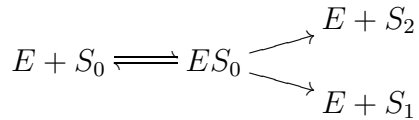
\begin{equation*}
\begin{split}
\xymatrix@R-2pc{
  & & E+S_2 \\ 
  E + S_0 \ar @<.4ex> @{-^>} [r] &  \ar @{-^>} [l] ES_0 \ar[ur] \ar[dr]  \\
  & & E+S_1 \\ 
}
\end{split}
\end{equation*}

Unfortunately, we can not extend our proof of Theorem \ref{thm:main-result} to establish stability of such networks because condition (G2) in Section~\ref{sec:assumptions} is violated.
\end{remark}

\begin{remark}
Feliu and Wiuf found that for small phosphorylation systems, including cascades, 
enzyme-sharing causes multistationarity~\cite{Enzyme-sharing}.  Our results give a partial converse: as long as the enzymes $E_i$ are {\em not} shared between rows of the all-encompassing network~\eqref{eq:m_network_processive}, then multistationarity is precluded.
\end{remark}

\begin{remark}

Only recently have there been studies of mixed phosphorylation mechanisms (partially distributive, partially processive)~\cite{Aoki}.
Suwanmajo and Krishnan proved that such a network, in which phosphorylation is distributive and dephosphorylation is processive (or, by symmetry, vice-versa), is {\em not} multistationary~\cite{SK}.  
Thus, it always admits a unique steady state, via a standard application of the Brouwer fixed-point theorem. 
This proves half of a conjecture that Conradi and Shiu posed~\cite{Shiu}.  

Perhaps surprisingly, the other half of the conjecture was essentially\footnote{The conjecture was stated for networks in which the processive mechanism is as in~\eqref{eq:multisite}, whereas oscillations were found in the network in which the processive mechanism is as in~\eqref{eq:irreversible}.  These networks  differ by only one reaction, when $n=2$, so it would be interesting to confirm whether adding this extra reaction, with small rate constant, also yields oscillations.} disproven by Suwanmajo and Krishnan: in contrast with processive systems (\S \ref{sec:pro-lit}), mixed systems need not be globally stable: they can be oscillatory~\cite{SK}!
\end{remark}

\section{Discussion} \label{sec:discussion}
Here we proved that a class of important biological networks -- fully processive phosphorylation/dephosphorylation cycles -- is globally convergent.  We did this by constructing an all-encompassing network that subsumes an infinite family of networks (e.g., reactions may be reversible or irreversible).  

Not only did this construction allow us to prove global convergence for many networks at once, but it also allowed us to incorporate network uncertainty into our analyses.  Indeed, one might not know whether specific reactions in a given biological network are reversible or irreversible, or whether one should incorporate product inhibition.  We therefore hope that our approach to handling network/model uncertainty may be useful in the future. To our knowledge, we are the first to introduce notation (\rev{}{}) to accommodate possibly reversible reactions.

We now return to the question from the start of this work: when are global dynamics preserved after adding or removing reactions and/or intermediate complexes from a network?  For the processive phosphorylation/dephosphorylation cycles in this work, we saw that the dynamics -- namely, global convergence to a unique equilibrium -- are preserved under these operations (where only the backward reaction may be added or removed in the context of \rev{}{}).  Many of these ideas came from~\cite{Freitas2}.

Does changing a reaction from reversible to irreversible always preserve global stability?  No.  For instance, the network $A \leftrightarrows B$ is globally stable (this can be checked by hand or by Proposition \ref{prop:summary-stability-results}), but $A \to B$ is not.

What about the opposite: {\em does changing a reaction from irreversible to reversible preserve global stability?}  We conjecture that this is false in general.

Finally, as noted earlier, monotone systems theory is only one of several approaches to proving global stability in reaction systems.  Which of these will allow us to prove more ``all-encompassing'' results?  In other words, which ones accommodate network uncertainty in the form of possibly reversible reactions and/or removing or adding intermediate complexes?  

\subsection*{Acknowledgments}
{\small 
ME conducted this research as part of the NSF-funded REU in the Department of Mathematics at Texas A\&M University (DMS-1460766), in which AS served as mentor.  The authors thank Muhammad Ali Al-Radhawi, Carsten Conradi, Michael Marcondes de Freitas, Shodhan Rao, and Robert Williams for helpful discussions.  The authors also thank two conscientious referees whose comments improved this work.  AS was supported by the NSF (DMS-1312473/DMS-1513364). 
} 
 
\begin{appendices}
\section{Proving Bounded-Persistence via Siphons} \label{sec:bounded_persistence}

Here we show bounded-persistence using P-semiflows and 
siphons~\cite{Freitas1,PetriNet}.
\begin{definition} Let $G= (\mathcal{S},\mathcal{C},\mathcal{R})$ be a reaction network with $s$ species and stoichiometric matrix $\Gamma$.
\begin{enumerate}
\item 
A \textbf{P-semiflow} (or {\em nonnegative conservation law}) of $G$ is any nonzero vector $v \in \R^{s}_{\geq 0}$ such that $\Gamma^T v = 0$. 
\item A nonempty subset of species $\Sigma \subseteq \mathcal{S}$ is a \textbf{siphon} of $G$ if every reaction of $G$ which has a product species in $\Sigma$ also has a reactant species in $\Sigma$. 
\item $G$ has the \textbf{siphon/P-semiflow property} 
if every siphon contains the support of a P-semiflow.
\end{enumerate}
\end{definition}
\begin{proposition} \label{thm:siphon}
Let $G$ be a reaction network 
that has the siphon/P-semiflow property, and let $\invtPoly$ be a stoichiometric compatibility class.  Then 
for all chemical reaction systems arising from $G$ and for all initial conditions $s_0 \in \invtPoly$, 
the {$\omega$-limit set} of the resulting trajectory does {\em not} intersect the boundary of $\invtPoly$, i.e., 
$\omega(s_0) \cap \partial \invtPoly = \emptyset$. 
Consequently, $G$ is bounded-persistent.
\end{proposition}
\begin{proof}
Let $s$ be the number of species. 
The first part follows from~\cite[Proposition 5.4]{PetriNet}, which states that the set of zero-coordinates of any $\omega$-limit point of a trajectory with initial condition in $\mathbb{R}^s_{\geq 0}$ is a siphon (if nonempty),
and~\cite[Lemma 3.4]{ShiuSturmfels}, which states that the siphon/P-semiflow property (labeled property ($\star$) there) is equivalent to the condition that {\em no} point in any compatibility class $\invtPoly$ has zero-coordinate set equal to a siphon.  
The ``Consequently'' part is immediate (we are considering initial conditions in $\mathbb{R}^s_{\geq 0}$ vs.\ $
\mathbb{R}^s_{> 0}$).
\end{proof}


\begin{lemma} \label{lem:bounded-persistent}
The 
all-encompassing network~\eqref{eq:m_network_processive} is bounded-persistent.  Moreover, for any stoichiometric compatibility class $\invtPoly$, any chemical reaction system arising from the network, and any initial condition $s_0 \in \invtPoly$, 
the resulting {$\omega$-limit set} does {\em not} intersect the boundary of $\invtPoly$, i.e., 
$\omega(s_0) \cap \partial \invtPoly = \emptyset$. 
\end{lemma}

\begin{proof}
By Proposition~\ref{thm:siphon}, we need only show that each siphon of network~\eqref{eq:m_network_processive} contains the support of a nonnegative conservation law (P-semiflow).  It is straightforward to check that each siphon contains the species (1) $E_i$, $C_{i1}$, $C_{i2}, \dots, C_{in_i}$ for some $i$, or (2) all $P_i$'s and all $C_{ij}$'s (for all $i,j$).  In the first case, the siphon contains the support of the conservation law for the total amount of free and bound enzyme $E_i$ (namely, $v \in \mathbb{R}^{2m+(n_1+\dots+n_m)}$ defined by $v_{P_i}=v_{C_{i1}}=\cdots = v_{C_{in_i}}=1$ and all others $=0$).  In the second case, the siphon contains the support of the conservation law for the total amount of free and bound substrate ($v_{P_i}=v_{C_{ij}}=1$ for all $i,j$ and all others $=0$).
\end{proof}

\end{appendices}

\section*{References}
\bibliographystyle{model1-num-names}
\bibliography{phos.bib}

\end{document}